\documentclass[runningheads]{llncs}

\input{src/preamble}

\author{
  Maurice~H.~ter~Beek\inst{1}\orcidID{0000-0002-2930-6367} \and
  Guillermina~Cledou\inst{2}\orcidID{0000-0003-0006-6440} \and
  Rolf~Hennicker\inst{3}\and
  Jos\'{e}~Proen\c{c}a\inst{4}\orcidID{0000-0003-0971-8919}
}

\institute{
  ISTI–CNR, Pisa, Italy,
    \email{maurice.terbeek@isti.cnr.it}
  \and
  HASLab, INESC TEC \& University of Minho, Portugal,
    \email{mgc@inesctec.pt}
  \and
  Ludwig-Maximilians-Universit\"{a}t M\"{u}nchen, Munich, Germany
    \and
  CISTER, ISEP, Polytechnic Institute of Porto, Portugal,
    \email{pro@isep.ipp.pt}
}

\authorrunning{ter Beek, Cledou, Hennicker, Proen\c{c}a}

\begin{document}

\title{Featured Team Automata}
\maketitle



\begin{abstract}
We propose featured team automata
to support variability in the development and analysis of teams, which are systems of reactive components that communicate according to specified synchronisation types.
A featured team automaton concisely describes a family of concrete product models
for specific configurations 
determined by feature selection.
We focus on the analysis of communication-safety properties, but doing so product-wise 
quickly becomes impractical.
Therefore, we investigate how to lift 
notions of receptiveness (no message loss)
to the level of family models.
%
We show that featured (weak) receptiveness of featured team automata
characterises 
(weak) receptiveness for all product instantiations.
A prototypical tool supports the developed theory.
\end{abstract}

\section{Introduction} 
\label{sec:introduction}


Team automata, 
originally introduced in the context of computer supported cooperative work 
to model groupware systems~\cite{Ell97}, 
are formalised as a theoretical framework to study  synchronisation mechanisms in system models~\cite{BEKR03}. Team automata represent an extension of I/O automata~\cite{BK05}. Their distinguishing feature 
is the loose nature of synchronisation according to which, in principle, any number of component automata can participate in the synchronised execution of a shared communicating action,
either as a sender or as a receiver. Team automata can determine specific synchronisation policies defining when and which actions are executed and by how many components.
%
Synchronisation types 
classify the 
policies realisable in team automata (e.g., peer-to-peer or broadcast communication) in terms of ranges for the number of sender and receiver components that can participate in a synchronisation~\cite{BCHK17}.
In extended team automata (\ETA)~\cite{BHK20b}, 
synchronisation type specifications (\STS) individually assign a synchronisation type to each communicating action. 
Such a specification uniquely determines a \emph{team} and 
gives rise to communication requirements to be satisfied by the team.

For systems composed by components communicating via message exchange, it is 
desirable to guarantee absence of 
communication failures, like message loss (typically output not received as input, violating receptiveness) or indefinite waiting (typically for input that never arrives, violating responsiveness). This requires knowledge of the synchronisation policies to establish the compatibility of communicating components~\cite{AH01,CC02,LNW07}; 
for team automata this was first studied for full synchronous products of component automata in~\cite{CK13}. 
Subsequently, a generic procedure to derive requirements for receptiveness and responsiveness for each synchronisation type was defined, and communication-safety of 
(extended) team automata
was expressed in terms of compliance with such 
requirements~\cite{BCHK17,BHK20b}. A team automaton 
is called compliant with a given set of communication requirements if in each reachable state 
the requirements are met (i.e.\ the communication is safe). If the required communication cannot occur immediately, but only after some arbitrary other actions have been executed, the team automaton is called weakly compliant
(akin to weak compatibility~\cite{BMSH10,HB18} or agreement of lazy request actions~\cite{BBDLFGD20}).
%

Many of today's software systems are highly configurable, variant-rich systems, developed as a software product line (SPL) with a notion of variability 
in terms of features that conceptualise pieces of system functionality or aspects that are relevant to the stakeholders~\cite{ABKS13}. Formal models of SPL behaviour are studied extensively. Such variability-rich behavioural models are often based on the super\-imposition of multiple product models in a single family model, equipped with feature-based variability such that each product model corresponds to a different configuration. Arguably the best known models are featured transition systems (\FTSs)~\cite{CHSLR10,Cla11,CCSHLR13}
and modal transition systems~\cite{FUB06,FG07}, possibly with variability constraints~\cite{ABFG11b,BFGM16}, but also I/O automata~\cite{LNW07,LPT09}, Petri nets~\cite{MPC11,MPC16} and contract automata~\cite{BGGDF17,BBDLFGD20} have been equipped with variability.
An \FTS is a labelled transition system (\LTS) whose transitions are 
annotated with feature expressions that are Boolean expressions over features, which condition the presence of transitions in product models, and a feature model, which 
determines the set of valid product models (configurations) of the family model.
The analysis of family models is challenging due to their innate variability, since the number of possible product models may be exponential in the number of features. 
In particular for larger models, enumerative product-by-product analysis becomes unfeasible; thus, dedicated family-based analysis techniques and tools, which exploit variability in terms of features, have been developed~\cite{CHSLR10,DS12,TSHA12,TAKSS14,DABW17,BVW17,CDKB18,CJJK19,BLVW20,Dim20}.

\paragraph{\rm\bf Motivation}

\FTSs have mostly been studied in the context of families of configurable components. Less attention has been paid on their parallel execution, in particular in the context of systems of reactive, concurrently running components,
where interaction is a crucial issue, often realised by message exchange.
For this, we need i)~to 
discriminate between senders and receivers and thus between input and output actions in \FTSs, and ii)~a flexible synchronisation mechanism, not necessarily peer-to-peer, for sets of \FTSs, called (featured) systems. In particular, the type of synchronisation should remain variable, depending on selected features (products).
Important questions for analysis of such systems concern behavioural compatibility (communication-safety).
As mentioned above, compositionality and communication-safety have been studied extensively in the literature for a variety of formal (automata-based) models, but---to the best of our knowledge---not considering variability.
Thus, we need a means to define and verify communication-safety for systems of \FTSs, ideally performing analyses on the level of featured systems such that the respective properties are automatically guaranteed for any product instantiation.
In this paper, we focus on the property of (weak) receptiveness.

\paragraph{\rm\bf Running Example} 

We consider a configurable access management system consisting of a server and users who can either login with secure authentication or without (open access). Concrete
automata 
capturing user and server behaviour are specified as 
family models whose product models correspond to configurations with or without secure authentication.

\cref{fig:fca-ex} shows two \FTSs: a family model of user components (\cref{fig:fca-ex-u}) and a family model of server components (\cref{fig:fca-ex-s}), as well as a feature model $\fm = \Fe{\lock \xor \unlock}$. 
The feature model expresses an exclusive choice of two features,
$\Fe{\lock}$ and $\Fe{\unlock}$, representing access with or without secure authentication, respectively, and defines two valid products (sets of features): $\{\Fe{\lock}\}$ and $\{\Fe{\unlock}\}$.
The idea is that the server must confirm login access only for secure authentication.
Thus, each transition is annotated with a constraint, denoted by a feature expression in square brackets (e.g., $\trFe{\lock[]}$), to indicate the product(s) that allow 
this transition.

A user starts in the initial state $0$, indicated by the incoming arrow, in which only the action \mi{join!} can be executed. Depending on the specific product, this results in a move to state $1$ (if feature \Fe{\lock} is present) or to state $2$ (if \Fe{\unlock} is present). From state $2$, a user can move back to state $0$ by executing action \mi{leave!}, in either product, as enabled by the transition constraint \trFe{\verum} (denoting truth value \kw{true}).
In state $1$, which is only present for 
the product with secure authentication, the user waits for explicit confirmation of login access from the server.


\cref{fig:projected-fca-u1,fig:projected-fca-u2} show the \LTSs representing the user product models, which result from projecting the user \FTS in \cref{fig:fca-ex-u} onto its set of valid products.
Similarly, \cref{fig:projected-fca-s1,fig:projected-fca-s2} show the \LTSs of the server product models, projecting the server \FTS onto its two valid products.

\begin{figure}[!hb]
\centering
\begin{subfigure}[c]{.55\textwidth}
\centering
\begin{tikzpicture}[align=center, every node/.style={scale=.9,font=\footnotesize}]
  \node[loc](0){$0$};
  \node[loc,right=2 of 0](2){$2$};
  \coordinate(mid) at ($(0)!.5!(2)$);
  \node[loc,above=0.3 of mid] (1) {$1$};
  \path[->] 
    (0) edge[bend left=25,above left] node{\trFe{\lock[]} \mi{join!}} (1)
    (0) edge[below] node{\trFe{\unlock[]} \mi{join!}} (2)
    (1) edge[bend left=25,above right] node{\trFe{\lock[]} \mi{confirm?}} (2)
    (2) edge[bend left=55,below] node{\trFe{\verum} \mi{leave!}} (0);
    \node[right=1.2 of 2] (fm) {\shl{$\fm = \Fe{\lock[\textcolor{black!80}]\xor \unlock[\textcolor{black!80}]}$}};
  \coordinate[left=0.3 of 0](initial);
  \draw[->,thick] (initial) -> (0);
\end{tikzpicture}
\vspace*{.75\baselineskip}
\caption{$\user$ser\qquad\qquad\qquad\qquad}
\label{fig:fca-ex-u}
\end{subfigure}
\quad
\begin{subfigure}[c]{.4\textwidth}
\centering
\begin{tikzpicture}[align=center, every node/.style={scale=.9,font=\footnotesize}]
  \node[loc](0){$0$};
  \node[loc,right=2 of 0](1){$1$};
  \coordinate(mid) at ($(0)!.5!(1)$);
  \path[->] 
    (0) edge[loop,out=135,in=225,looseness=5,left] node{\trFe{\unlock[]} \mi{join?}} (0)
    (0) edge[loop,out=45,in=-45,looseness=5,right] node{\trFe{\verum} \mi{leave?}} (0)
    (0) edge[bend left=55,above] node{\trFe{\lock[]} \mi{join?}} (1)
    (1) edge[bend left=55,below] node{\trFe{\lock[]} \mi{confirm!}} (0);
  \coordinate[above=0.3 of 0](initial);
  \draw[->,thick] (initial) -> (0);
\end{tikzpicture}
\caption{$\server$erver}
\label{fig:fca-ex-s}
\end{subfigure}
\caption{Family models of users $\user$ and servers $\server$ and a shared feature model \fm}
\label{fig:fca-ex}
\centering
\begin{subfigure}[t]{.35\textwidth}
\centering
\begin{tikzpicture}[align=center, every node/.style={scale=.9,font=\footnotesize}]
  \node[loc](0){$0$};
  \node[loc,right=1.5 of 0](2){$2$};
  \coordinate(mid) at ($(0)!.5!(2)$);
  \node[loc,above=0.3 of mid] (1) {$1$};
  \path[->] 
    (0) edge[bend left,above left] node{\mi{join!}} (1)
    (1) edge[bend left,above right] node{\mi{confirm?}} (2)
    (2) edge[below] node{\mi{leave!}} (0);
  \coordinate[left=0.3 of 0](initial);
  \draw[->,thick] (initial) -> (0);
\end{tikzpicture}
\caption{$\user\proj_\Fe{\lock}$}
\label{fig:projected-fca-u1}
\end{subfigure}
\begin{subfigure}[t]{.25\textwidth}
\begin{tikzpicture}[align=center, every node/.style={scale=.9,font=\footnotesize}]
  \node[loc](0){$0$};
  \node[loc,right=1.5 of 0](1){$1$};
  \path[->] 
    (0) edge[loop,out=35,in=-35,looseness=4,right] node{\mi{leave?}} (0)
    (0) edge[bend left=45,above] node{\mi{join?}} (1)
    (1) edge[bend left=45,below] node{\mi{confirm!}} (0);
  \coordinate[above=0.3 of 0](initial);
  \draw[->,thick] (initial) -> (0);
\end{tikzpicture}
\caption{$\server\proj_\Fe{\lock}$}
\label{fig:projected-fca-s1}
\end{subfigure}
\begin{subfigure}[t]{.25\textwidth}
\centering
\begin{tikzpicture}[align=center, every node/.style={scale=.9,font=\footnotesize}]
  \node[loc](0){$0$};
  \node[loc,right=1.5 of 0](2){$2$};
  \path[->] 
    (0) edge[below] node{\mi{join!}} (2)
    (2) edge[bend left=55,below] node{\mi{leave!}} (0);
  \coordinate[above=0.3 of 0](initial);
  \draw[->,thick] (initial) -> (0);
\end{tikzpicture}
\caption{$\user\proj_\Fe{\unlock}$}
\label{fig:projected-fca-u2}
\end{subfigure}
\begin{subfigure}[t]{.125\textwidth}
\begin{tikzpicture}[align=center, every node/.style={scale=.9,font=\footnotesize}]
  \node[loc](0){$0$};
  \path[->] 
    (0) edge[loop,out=45,in=135,looseness=5,above] node{\mi{join?}} (0)
    (0) edge[loop,out=225,in=-45,looseness=5,below] node{\mi{leave?}} (0);
  \coordinate[left=0.3 of 0](initial);
  \draw[->,thick] (initial) -> (0);
\end{tikzpicture}
\caption{$\server\proj_\Fe{\unlock}$}
\label{fig:projected-fca-s2}
\end{subfigure}
\caption{Product models of users and servers (projections of the models in \cref{fig:fca-ex})
}
\label{fig:projected-fca}
\end{figure}

\begin{figure}[!ht]
	\centering
	\begin{tikzpicture}[x=2.50cm, y=-1.4cm]
		\tikzstyle{obj} = [inner sep=0pt, outer sep=1mm,font=\small]
		\tikzstyle{arr} = [inner sep=0pt, -stealth]
		\tikzstyle{lbl} = [inner sep=2pt, font=\scriptsize]
    \tikzstyle{slbl}=[below,inner sep=8pt,font=\sf\smaller]
    \tikzstyle{sqgl}=[arr,decoration={snake,amplitude=0.7pt,segment length=1.2mm,
                      post=lineto,post length=3pt},decorate] 

		\node [obj] (FCA) at (0.05,0) {$A\,{:}\,\mFCA$};
		\node [obj] (FS) at (1.1,0) {$\S {=} (A_i)_{i\in\N}{:}\,\textit{\FSys}$ };
		\draw [arr] (FCA) to node [lbl, above=1.0mm] {} (FS);
		
		\node [obj] (CA) at (0.05,1) {$A\proj_p\,{:}\,\mi{\CA}$};
		\node [obj] (S) at (CA-|FS) {$\S\proj_p {=}(A_i\proj_p)_{i\in\N}{:}\, \textit{Sys}$ };
		\draw [arr] (CA) to node [lbl, above=1.0mm] {} (S);
		
		\node [obj] (FETA) at (2.75,0) {$\mFETA\times f\ensuremath{\mkern-3mu}Reqs\xspace$};
		\node [obj,xshift=-0pt] (ETA) at (CA-|FETA) {$\mi{\ETA}\times\REQs$};

    \node [obj,right=0.2 of FETA,align=center] (frec)
      {\smaller featured\\[-4pt]\smaller (weakly)\\[-4pt]\smaller receptive};
    \node [obj,align=center]at(ETA-|frec) (rec)
      {\smaller (weakly)\\[-4pt]\smaller receptive};
    \draw [stealth-stealth,double] (frec) to (rec);
    \draw [slbl,inner sep=2pt,draw=none]
          (frec.center) to node[yshift=0pt,right]{Thm.\,\ref{thm:frec-to-rec}} (rec.center);

		\draw[arr] (FS) to node[lbl,below,xshift=3pt] (fst) {$\fstype{:}\mFSTS$} (FETA);
		\draw[arr] (S)  to node[lbl,above,xshift=-2pt] (st) {$\stype{:}\textit{\STS}$}  (ETA);

		\draw[arr] (FCA)  to node[lbl,right,xshift=-1pt] {$\proj_p$}  (CA);
		\draw[arr] (FS)   to node[lbl,right,xshift=-1pt] {$\proj_p$}  (S);
		\coordinate (fetasw) at ($(FETA.south west)+(15pt,0)$);
		\draw[arr] (fetasw) to node[lbl,right,xshift=-1pt] {$\proj_p$}  (ETA.north-|fetasw);
		\draw[arr] (fst)  to node[lbl,right,xshift=-1pt] {$\proj_p$}  (st.north-|fst);

    \draw[sqgl] (FETA) to (frec);
    \draw[sqgl] (ETA) to (rec);

    \begin{scope}[on background layer]
      \tikzstyle{bg}=[draw=none]
      \coordinate (case) at ($(CA.south east)-(2pt,0)$);
      \coordinate (fetanw) at ($(FETA.north west)-(4pt,0)$);
      \coordinate (etaL) at ($(ETA.center)-(8pt,0)$);
      \coordinate (etaR) at ($(ETA.center)+(6pt,0)$);
      \coordinate (stL) at ($(st.west)+(1pt,0)$);
      \coordinate (SR) at ($(FS.east)+(2pt,0)$);
      \coordinate (SL) at ($(FS.west)+(-2pt,0)$);
      \coordinate (frecR) at ($(frec.east)+(6pt,0)$);
      \coordinate (recR) at ($(rec.east)+(4pt,0)$);
      \node[fit=(FCA.north west|-FETA.north)(case|-ETA.south),
            fill=yellow!30,draw=yellow!75!gray,bg](comp){};
      \node[fit=(FETA.north-|S.center)(SL)(SR)(ETA.south-|S.center),
            fill=orange!18,draw=orange!75!gray,bg](sys){};
      \node[fit=(stL)(fetanw)(ETA.south west)(etaL),
            fill=red!15,draw=red!75!gray,bg](team){};
      \node[fit=(etaR)(FETA.north east|-FETA.north)(frecR)(rec.south east|-ETA.south),
            fill=purple!15,draw=purple!75!gray,bg](req){};
      \node[slbl]at(comp.south){Components (§\ref{sub:fca})};
      \node[slbl]at(sys.south){Systems (§\ref{sub:fsystems})};
      \node[slbl]at(team.south){Teams (§\ref{sub:featured_eta},\,Thm.\,\ref{thm:commuting})};
      \node[slbl]at(req.south){Receptiveness (§\ref{sec:receptiveness})};
      \node[fit=(FCA)(FETA)(frecR),inner sep=4pt,draw=black!30,thick,dotted]{};
      \node[fit=(CA.north-|FCA.west)(ETA.south)(recR),inner sep=4pt,draw=black!30,thick,dotted]{};
      \node[lbl,left,align=right](other)at(CA.west){Previous\,\\work~\cite{BHK20b}:\,};
      \node[lbl,left,align=right](this)at(other.east|-FCA){Current\,\\paper:\,};
    \end{scope}

	\end{tikzpicture}
	\vspace*{-1.5\baselineskip}
	\caption{\label{fig:overview}Overview of this paper, using a valid product $p$}
 
\vspace*{-\baselineskip}
\end{figure}
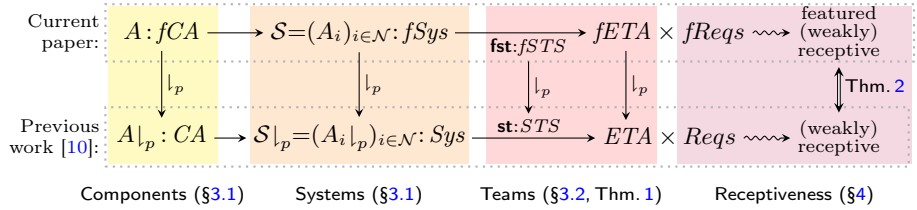

\paragraph{\rm\bf Contribution} 

\cref{fig:overview} 
illustrates the contents and contributions of this paper, which we now explain and relate to the literature mentioned above. In particular, we 
extend~\cite{BHK20b}, by enriching \ETA with variability, proposing a new model called \emph{featured} ETA (\FETA) to allow the specification of---and reasoning on---a family of \ETA parameterised by a set of features. We define \emph{projections}
$\proj_p$ (for any valid product $p$) to relate the featured setting of this paper to that without in~\cite{BHK20b}.

First, we extend \emph{component automata} (\CA), the building blocks of (extended) team automata, with variability, resulting in \FTSs with input and output actions, called \emph{featured} \CA (\FCA). Basically, \CA are \LTSs that distinguish between input and output actions (and internal actions, omitted in this paper) and which capture the behaviour of a component. The \FTSs in the running example are \FCA in which input and output actions
are appended by $!$ and $?$, respectively.
Multiple \CA can run in parallel to form a \emph{system} (\textit{Sys} in \cref{fig:overview}) of the \CA; we propose a \emph{featured} \Sys (\FSys) to consist of 
\FCA instead of \CA.

Given a \Sys and a \emph{synchronisation type specification} (\STS), it is possible to generate an \ETA and derive \emph{receptiveness requirements} (\REQs), and study whether the \ETA is (weakly) compliant with all such \REQs, in which case it is called \emph{(weakly) receptive}.
An \ETA is an \LTS that restricts how \CA in the \Sys can communicate based on the \STS.
We propose a \emph{featured} \STS (\FSTS) to parameterise an \STS with variability, giving rise to the aforementioned \FETA and \emph{featured} \REQs (\FREQs).
If the \FETA is featured (weakly) compliant with all such \FREQs, it is called \emph{featured (weakly) receptive}.

While the extension from \CA to \FCA (and from systems to featured systems) is rather straightforward, \FETA are not simple extensions of \ETA: the \FSTSs giving rise to \FETA are a nontrivial extension of the \STSs for \ETA, partially due to the variability in synchronisation types.
Our first result (\Cref{thm:commuting}) confirms the soundness of our extension.
Our main result (\Cref{thm:frec-to-rec}) is that featured (weak) receptiveness induces and reflects
(weak) receptiveness of product models, i.e.\
a \FETA is featured (weakly) receptive if and only if
all \ETA obtained by product projections are (weakly) receptive. 

\paragraph{\rm\bf Outline}
%
\cref{sec:preliminaries} provides some basic definitions concerning variability. 
\cref{sec:featured_team_automata} lifts the theory of team automata to that of featured team automata, and \cref{sec:receptiveness} does the same for receptiveness 
requirements and compliance. We present a prototypical implementation of the developed theory in \cref{sec:implementation}, and \cref{sec:conclusion_and_discussion} concludes the paper 
and provides some ideas for future work.
%
The proofs of our results can be found in~\cref{sec:proofs}. 

\section{Variability} 
\label{sec:preliminaries}

This section provides 
definitions of the basic notions concerning variability, viz.\ 
\emph{features}, \emph{feature expressions}, \emph{feature models}, 
and 
\FTSs. 


A \emph{feature}, ranged over by $f$, is regarded as a Boolean variable that represents a unit of 
variability. This paper assumes a finite set of features $F$.
A \emph{product}, ranged over by $p\subseteq F$, is a finite subset of selected features.
In the context of SPLs, 
a product can be interpreted as a configuration used to derive concrete software systems.
A \emph{feature expression} $\psi$ over a set of features $F$, denoted $\psi\in \mi{FE}(F)$, 
is a Boolean expression over features with the usual Boolean connectives and constants \verum and \falsum interpreted by the truth values \kw{true} and \kw{false}.
A product $p$ satisfies a feature expression $\psi$, denoted $p\models \psi$, if and only if $\psi$ is evaluated to \verum
if \verum is assigned to every feature in $p$ and \falsum to the features not in $p$. A feature expression $\psi$ is \emph{satisfiable} if there exists a product $p$
such that $p\models \psi$.
A \emph{feature model} $\fm\in \mi{FE}(F)$ is a feature expression that determines 
the set of products for which concrete systems of an SPL
can be derived.
We use $\sem{\fm}$
to denote the 
set of products that satisfy the feature model $\fm\in \mi{FE(F)}$. 


\myparagraph{Notation.}
For any product 
$p \subseteq F$, 
its view as a feature expression is 
$\Chi{p} = \bigwedge_{f\in p} f \land \bigwedge_{f\in F\backslash p} \lnot f$.
$p$ is the unique product with $p\models \Chi{p}$.
A set $P$ of products is characterised by the feature expression
$\Chi{P} = \bigvee_{p\in P} \Chi{p}$. Clearly, for any product $p$, 
$p \in P$ iff $p \models \Chi{P}$. Note that the conjunctions and disjunctions are finite, since $F$ is finite. Moreover, $\bigwedge_{i \in \emptyset}\psi_i$ stands for $\verum$ and $\bigvee_{i \in \emptyset}\psi_i$ stands for $\falsum$.

\smallskip


A \emph{featured transition system} (\FTS) is a tuple $A\,{=}\,(Q,I,\Sigma,E,F,\fm,\gamma)$
such that $(Q,I,\Sigma,E)$ is an \LTS 
with  
a finite set of states $Q$, 
a set of initial states $I\subseteq Q$,
a finite set of actions $\Sigma$, and  
a transition relation $E\subseteq Q\times \Sigma \times Q$.
$F$ is a finite set of features, $\fm \in \FE(F)$ is a feature model
and $\gamma: E\to \FE(F)$ is a mapping assigning feature expressions to
transitions. 
A product $p \subseteq F$ is \emph{valid} for the feature model $\fm$,
if $p \in \sem{\fm}$.
The mapping $\gamma$ expresses \emph{transition constraints} for the realisation of transitions. 
A transition $t \in E$ is \emph{realisable} for a valid product $p$
if $p \models \gamma(t)$.

An \FTS $A$ can be \emph{projected} to a valid product $p$
by using $\gamma$ to filter realisable transitions,
resulting in
the LTS $A\proj_p = (Q,I,\Sigma,E\proj_p)$, where
$E\proj_p = \{t\in E \mid$ $p \models \gamma(t)\}$.
Such a projection is also called product model or configuration.
Hereafter, we will generally write projections using superscripts, e.g. $A^p$ to denote~$A\proj_p$.
%


\myparagraph{Notation.}
Given an LTS or an \FTS $A$, we write 
$q\tr{a}_A q'$, or shortly $q\tr a q'$, to denote $(q,a,q') \in E$. 
For $\Gamma\subseteq\Sigma$, we write $q\tr{\Gamma}\!\!^{*}\,q'$ if there exist 
$q\tr{a_1}q_1\tr{a_2}\cdots\tr{a_n}q'$ 
for some $n\geq0$ and $a_1,\dots,a_n\in \Gamma$.
%
%
An action $a$ is \emph{enabled} in $A$ at state $q\in Q$, 
denoted $a\,\bkw{en}_A@q$, if there exists $q'\in Q$ such that $q \tr{a} q'$. 
A state $q\in Q$ is \emph{reachable} if
$q_0\tr{\Sigma}\!\!^*\,q$ for some $q_0 \in I$.



\section{Team Automata with Variability} 
\label{sec:featured_team_automata}

This section proposes 
to integrate variability in the
modelling of teams of reactive components which communicate according
to specified synchronisation policies.
For this purpose we define 
\emph{featured \CA}, \emph{featured \Syss}, and \emph{featured \ETA}, and
compare them to their featureless counterparts.
Throughout this section we will use \shl{grey backgrounds} to highlight extensions with features.


\subsection{Featured Component Automata and Featured Systems} 
\label{sub:fca}\label{sub:fsystems}

A \emph{featured component automaton} (\FCA) is an \FTS $A = (Q,I,\Sigma,E,\mshl{F},\mshl{\fm},\mshl{\gamma})$
such that $\Sigma=\Sigma^?\uplus\Sigma^!$ 
consists of 
disjoint sets $\Sigma^?$ of \emph{input actions}
and $\Sigma^!$ of \emph{output actions}. For simplicity, we do not consider
internal actions here.
For easier readability, input actions will be shown 
with suffix ``$?$'' and output actions with suffix ``$!$''. 
\FCA extend \emph{component automata} (\CA)~\cite{BEKR03,BHK20b} with features and feature models.
The running example in~\Cref{sec:introduction} contains examples of \FCA.

A \emph{featured system} (\FSys) is a pair $\S = (\N,(A_i)_{i\in \N})$, where 
$\N$ is a finite, nonempty set of component names and 
$(A_i)_{i\in \N}$ is an $\N$-indexed family of \FCA $A_i = (Q_i,I_i,\Sigma_i,E_i,\mshl{F},\mshl{\fm},\mshl{\gamma_i})$ over a shared set of features $\mshl{F}$ and feature model $\mshl{\fm}$. 
Composition of feature models is out of the scope of this paper, but note that multiple approaches exist in the literature, e.g., using conjunction or disjunction of feature models~\cite{SHTB06,Cla11,CCSHLR13}.

%
Featured systems extend \emph{systems of \CA}~\cite{BHK20b} by using \FCA instead of \CA as system components.
%
%
An \FSys 
$\S = (\N,(A_i)_{i\in \N})$ induces:
the set of system states $Q=\prod_{i\in\N} Q_i$ 
such that, for any $q\in Q$ and for all $i\in\N$, $q_i\in Q_i$; 
the set of initial states $I=\prod_{i\in\N} I_i$;
the set of system actions $\Sigma=\bigcup_{i\in\N} \Sigma_i$;
%
%
the set of system labels
$\Lambda \subseteq \pset{\N}\times \Sigma \times \pset{\N}$
defined as $\Lambda = \{(S,a,R) \mid \emptyset \neq S\cup R \subseteq \N,\allowbreak
\forall_{i\in S}\cdot a\in \Sigma^!_i,
\forall_{i\in R}\cdot a\in \Sigma^?_i\}$;
and the set of system transitions
$E \subseteq Q\times \Lambda \times Q$ 
defined as 
$E = \{q\tr{(S,a,R)}q' \mid 
\forall_{i\in(S\cup R)} \cdot q_i\tr{a}_{A_i} q'_i, ~
\forall_{j\in\N\setminus (S\cup R)} \cdot q_j=q_j' \}$.
%
%
%

A transition labelled by a \emph{system label} denotes the atomic execution of an action $a$ by a set of components in which $a$ is enabled.
%
More concretely,
for a system label $(S,a,R)\in\Lambda$,
$S$ represents the set of \emph{senders} and $R$ the set of \emph{receivers} that synchronise on an action $a \in \Sigma$.
Since, by definition of system labels, $S\cup R \neq \emptyset$, at least
one component participates in any system transition.
The transitions of a 
system capture 
all possible synchronisations of shared actions of its components,
even when only one component participates.
Given a system transition
$t = q\tr{(S,a,R)}q'$, we write $t.a$ for $a$, $t.S$ for $S$
and $t.R$ for $R$.
For ease of presentation, we assume in this paper that systems are closed.
This means that any system action $a \in \Sigma$
occurs in (at least) one of its components as an input action and in (at least) one of its components as an output action.
%
%

%


The \emph{projection} of an \FSys $\S = (\N,(A_i)_{i\in \N})$
to a product $p \in \sem{\fm}$ is the system
$\S^p = (\N,(A_i^{p})_{i\in \N})$. 

\begin{example}
\label{ex:fsys}
%
We consider an \FSys ${\S_\auth}$ with three components, two users and one server following the running example in~\Cref{sec:introduction}. Formally,
${\S_\auth} = (\mc{N}, (A_i)_{i\in\mc{N}})$, where
$\mc{N} =\set{u_1,u_2,s}$ are component names,
$A_{u_1},A_{u_2}$ are copies of the
\FCA $\user$ in~\cref{fig:fca-ex-u},
and $A_s$ is a copy of the \FCA $\server$ in \cref{fig:fca-ex-s}.


The system states are tuples $(p,q,r)$ 
with user states $p \in Q_{u_1}$ and  $q \in Q_{u_2}$, and server state $r\in Q_{s}$.
$\S_\auth$ has an initial state $(0,0,0)$, 
a total of 18~states ($3\times 3\times 2$),
actions $\com=\{\mi{join},\mi{leave},\mi{confirm}\}$, and a total of 142 system transitions.
Some of these (with action $\mi{leave}$) are depicted in \cref{fig:fsys}; the transitions \textcolor{black!50}{marked in grey} will be discarded based on synchronisation restrictions in the next section.



The projection of ${\S_\auth}$ to the valid product $\{\Fe{\lock}\}$, respecting 
the shared feature model \Fe{\lock \xor\unlock},
is the system $\S_\auth^\Fe{\lock} = 
(\mc{N}, \set{A_{u_1}^\Fe{\lock},A_{u_2}^\Fe{\lock},A_{s}^\Fe{\lock}})$,
such that
$A_{u_1}^{\Fe{\lock}},A_{u_2}^{\Fe{\lock}}$
are copies of 
$\user\proj_\Fe{\lock}$
in~\cref{fig:projected-fca-u1} and 
$A_{s}^{\Fe{\lock}}$ is a copy of 
$\server\proj_\Fe{\lock}$
in~\cref{fig:projected-fca-s1}.
Similarly, for product $\set{\Fe{\unlock}}$,
we get the projected system ${\S_\auth^{\Fe{\unlock}}} = 
(\mc{N}, \set{A_{u_1}^{\Fe{\unlock}},A_{u_2}^{\Fe{\unlock}},A_{s}^{\Fe{\unlock}}})$.
%
%
\qedx
\end{example} 
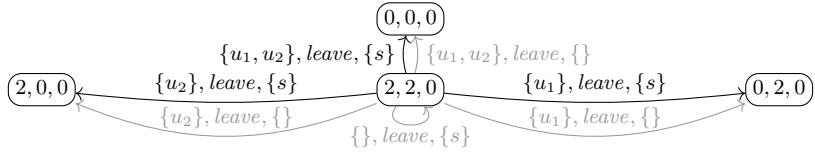
\begin{figure}[t]
\centering
\begin{tikzpicture}[align=center, every node/.style={scale=.9,font=\footnotesize}]
  \tikzstyle{fade}=[draw=black!40,color=black!40]
  \node[rloc](220){$2,2,0$};
  \node[rloc,right=4 of 220](020){$0,2,0$};
  \node[rloc,left=4 of 220](200){$2,0,0$};
  \node[rloc,above=0.5 of 220](000){$0,0,0$};
  \path[->] 
    (220)
       edge[bend left=6,above] node{\Label{u_2}{leave}{s}} (200)
      edge[fade,bend left=22,above] node{\Label{u_2}{leave}{}} (200)
      edge[bend right=6,above] node{\Label{u_1}{leave}{s}} (020)
      edge[fade,bend right=22,above] node{\Label{u_1}{leave}{}} (020)
      edge[     bend left=15,left]   node{\Label{u_1,u_2}{leave}{s}} (000)
      edge[fade,bend right=15,right] node{\Label{u_1,u_2}{leave}{}}  (000)
      edge[fade,loop,out=230,in=310,looseness=3,below]
        node[yshift=3pt,xshift=0pt](loop){\Label{}{leave}{s}} (220)
  ;
%
%
%
  %
  \pgfresetboundingbox
  \node[inner sep=0]at(000.north) {};
  \node[inner sep=0]at(loop.south) {};
  \node[inner sep=0]at(020.west) {};
  \node[inner sep=0]at(200.east) {};
\end{tikzpicture}
\caption{Some system transitions of $\S_\auth$}
\label{fig:fsys}
\end{figure}



\subsection{Featured Team Automata} 
\label{sub:featured_eta}




\emph{Featured team automata} (\FETA) are the key concept to model families
of teams. They are constructed over an \FSys \S together with a specification of synchronisation types expressing desirable synchronisation constraints. This section first formalises the latter and then \FETA as \FTSs.

A \emph{synchronisation type} $(s,r)\!\in\!\kw{Intv}{\times}\kw{Intv}$ is a pair of intervals $s$ and $r$ which de-\linebreak termine the number 
of senders and receivers that can participate in a communication.
Each interval 
is written $[\mi{min},\mi{max}]$,
with $\mi{min}\!\in\!\Nat$ and $\mi{max}\!\in\!\Nat \cup \{*\}$. 
We use $*$ to denote $0$ or more participants, and write 
$x\in[n,m]$ if $n\!\leq\!x\leq\!m$ and $x\in[n,*]$ if $x\geq n$.
For a system transition $t$, we define
$t \models (s,r)$ if
$|t.S|\in s \land |t.R|\in r$.

A \emph{featured synchronisation type specification} (\FSTS) 
over an \FSys 
\S, is a total function, 
$\fstype\!:\mshl{\sem{\fm}} \times \com \to \kw{Intv}\!\times\!\kw{Intv}$, mapping each product $p\in\sem{\fm}$ and 
action $a\in\com$ to a synchronisation type.
Thus, an \FSTS is parameterised by (valid) products and therefore
supports variability of synchronisation conditions.

\FSTSs are extensions of synchronisation type specifications (\STSs) in~\cite{BHK20b};
an \STS $\stype: \com \to \kw{Intv}\times\kw{Intv}$ maps actions to bounds of senders and receivers.
For any product $p\in \sem{\fm}$,
an \FSTS \fstype can be projected to an \STS $\fstype^p$
such that $\fstype^p(a) = \fstype(p,a)$ for all $a \in \Sigma$.

\begin{example}
\label{ex:fsts-auth}
The 
definition of $\fstype_\auth$ corresponds to an \FSTS 
for the \FSys ${\S_\auth}$ in \cref{ex:fsys}:
\vspace*{-.75\baselineskip}
\begin{align}
\fstype_{{\auth}}(p,\mi{confirm}) &= ([1,1],[1,1]) \text{ for }p\in\set{\{\Fe{\lock}\},\{\Fe{\unlock}\}}\label{eq:fst1}\\
\fstype_{{\auth}}(\set{\Fe{\lock}},a)    &= ([1,1],[1,1]) \text{ for } a\in \set{\mi{join},\mi{leave}}\label{eq:fst2} \\
\fstype_{{\auth}}(\set{\Fe{\unlock}},a)    &= ([1,*],[1,1]) \text{ for } a\in \set{\mi{join},\mi{leave}}\label{eq:fst3}
\end{align}

\vspace*{-.25\baselineskip}
\noindent Intuitively, independently of the selected product, users can receive \emph{confirm}ation
from the server in a \emph{one-to-one} fashion~(\ref{eq:fst1}). 
If secure authentication~\Fe{\lock} 
is required, 
one user can \emph{join}/\emph{leave} by synchronising exclusively 
with one server~(\ref{eq:fst2}).
If open access 
\Fe{\unlock} is required, 
multiple users can \emph{join}/\emph{leave} at the same time~(\ref{eq:fst3}).
\qedx
\end{example}
Given an \FSys $\S = (\N,(A_i)_{i\in \N})$ and an \FSTS \fstype over \S,
the \emph{featured team automaton} (\FETA) generated by $\S$ and \fstype, written $\fstype[\S]$, is the
\FTS $(Q,\ab I,\ab \Sigma,\ab E,\ab \mshl{F},\ab \mshl{\fm},\ab \mshl{\gamma})$
where $Q,I,\Sigma,\allowbreak E,\mshl{F}$, and $\mshl{\fm}$ 
are determined by $\S$.
It remains to construct the mapping $\gamma: E \to \mi{FE}(F)$,
which constrains system transitions by feature expressions.
The definition of $\gamma$ is derived from both the transition constraints
$\gamma_i$ of every $A_i$ and from \fstype.
It is motivated by the fact that a system transition $t = q\tr{(S,a,R)}q' \in E$
should be realisable for those products $p\in\sem{\fm}$ for which
both of the following conditions hold:

\begin{enumerate}
    \item 
    In each component $A_i$, with $i \in (S\cup R)$,
the local transition $q_i\tr{a}_{A_i} q_i'$ is realisable
for $p$. This means $p \models \hat{\gamma}(t)$, where 
$\hat{\gamma}(t) = \bigwedge_{i\in (S\cup R)} \gamma_i(q_i\tr{a}_{A_i} q_i')$.
\item
For any action 
$a \in \com$, 
the number of senders $|S|$ and receivers $|R|$ 
fits the synchronisation type $\fstype(p,a)$.
This means $p \models \Chi{P(\fstype,t)}$, where $\Chi{P(\fstype,t)}$
(cf.~\cref{sec:preliminaries}) 
is the feature expression characterising the set of products
$P(\fstype,t) = \{p\in\sem{\fm} \mid\allowbreak
t\models \fstype(p,t.a) \}$.
\end{enumerate}

\noindent In summary, for any $t  = q\tr{(S,a,R)}q' \in E$, we define
$\gamma(t) = \hat{\gamma}(t)\land \Chi{P(\fstype,t)}$.
 Note that, since $P(\fstype,t)$ is a subset of $\sem{\fm}$,
it holds
$\models \Chi{P(\fstype,t)} \rightarrow \fm$ and hence $\models \gamma(t) \rightarrow \fm$. In cases where $P(\fstype,t) = \sem{\fm}$, $\Chi{P(\fstype,t)}$ and  $\fm$ are equivalent and then we will often use
$\gamma(t) = \hat{\gamma}(t) \land \fm$.

\smallskip

Recall that an \FTS can be projected to products (as defined in~\cref{sec:preliminaries})
and therefore also the \FETA $\fstype[\S]$ can be projected 
to a valid product $p \in \sem{\fm}$ yielding the LTS
$\fstype[\S]^p$. Thus any \FETA $\fstype[\S]$
specifies a family of product~models.


\begin{example}\label{ex:feta}
Consider the \FSys ${\S_\auth}$ and the \FSTS
$\fstype_\auth$ from~\cref{ex:fsts-auth}, here and in the following examples simply called \fstype,
as well as the generated \FETA $\fstype[\S_\auth]$.
There are many system transitions, for instance
\begin{align*}
t_1 &=
(0,0,0) \tr{\mathit{(\{u_1,u_2\},\,join,\{s\})}} (2,2,0) \text{ and}\\
t_2 &=
(0,0,0) \tr{(\{u_1,u_2\},\,\mathit{join},\{s\})} (1,1,1).
\end{align*}
For $t_1$, we have 
$\hat{\gamma}(t_1) = 
	\bigwedge_{i\in\{1,2\}} \gamma_{u_i}(0 \tr{\mathit{join}}_{A_{u_i}} 2) \land
	\gamma_{s}(0 \tr{\mathit{join}}_{A_s} 0)=\Fe{\unlock}\land\unlock\land\unlock$. 
Since $\{\Fe{\unlock}\}$ is the only valid product $p$ such that
$t_1 \models \fstype(p,\mathit{join}) = ([1,*],[1,1])$---note that
only for open access 
more than one user can join simultaneously---we have $\Chi{P(\fstype,t_1)} = \Fe{\unlock} \land\ab \neg\,\Fe{\lock}$
(where $P(\fstype,t_1) = \{\{\Fe{\unlock}\}\}$).
Thus, in summary, $\gamma(t_1) = (\Fe{\unlock} \land\ab \Fe{\unlock} \land\ab \Fe{\unlock}) \land\ab (\Fe{\unlock} \land\ab \neg\,\Fe{\lock})$.
Hence $t_1$ can only be realised for open access. 

For $t_2$, we have $\hat{\gamma}(t_2) = \Fe{\lock} \land\ab \Fe{\lock} \land\ab \Fe{\lock}$
and $\Chi{P(\fstype,t_2)} = \Fe{\unlock} \land\ab \neg\,\Fe{\lock}$ as before,
since $\{\Fe{\unlock}\}$ is the only product $p$ such that
$t_2 \models \fstype(p,\mathit{join})$.
Therefore, $\gamma(t_2) = (\Fe{\lock} \land\ab \Fe{\lock} \land\ab \Fe{\lock})
\land\ab (\Fe{\unlock} \land\ab \neg\,\Fe{\lock})$, which reduces
to \Fe{\bot} and thus is not realisable by any product.

\cref{fig:feta} shows the full generated \FETA 
$\fstype[{\S_\auth}]$,
after removing all 
unreachable states and all 
non-realisable transitions $t$,
i.e.\ $\forall_{p\in\sem{\fm}} \cdot  p\not\models\gamma(t)$.
For each transition $t$ in~\cref{fig:feta} we present
$\gamma(t)$ as a conjunction of (a semantics-preserving simplification of)
$\hat{\gamma}(t)$ and an \underline{underlined} $\Chi{P(\fstype,t)}$
or $\fm = \Fe{\lock}\xor\Fe{\unlock}$ if $P(\fstype,t) = \sem{\fm}$.
The latter is the case in all transitions in which only one user participates.
If two users join or leave simultaneously, then
$\Chi{P(\fstype,t)}$ is always $\Fe{\unlock}\land\!\neg\,\Fe{\lock}$
as explained above for $t_1$. (Further reductions 
are possible for the conjoined $\gamma(t)$.)
\qedx
\end{example}

\begin{figure}[tb!]
\centering

\begin{tikzpicture}[align=center, every node/.style={scale=.9,font=\footnotesize}]
  \node[rloc](000){$0,0,0$};
  \node[rloc,right=5 of 000](020){$0,2,0$};
  \coordinate(mid1) at ($(000)!.5!(020)$);
  \node[rloc,above=1 of mid1](011){$0,1,1$};
  \node[rloc,below=2.5 of 000](200){$2,0,0$};
  \node[rloc,below=2.5 of 020](220){$2,2,0$};
  \coordinate(mid2) at ($(200)!.5!(220)$);
  \node[rloc,below=1 of mid2](211){$2,1,1$};
  \coordinate(mid3) at ($(020)!.5!(220)$);
  \node[rloc,right=1 of mid3](121){$1,2,1$};
  \coordinate(mid4) at ($(000)!.5!(200)$);
  \node[rloc,left=1 of mid4](101){$1,0,1$};
  \path[->] 
    (000) edge[bend left=7,above] 
          node{\trFe{\unlock[]\land\ul{\fm}}\Label{u_2}{join}{s}} 
          (020)
    (000.north) edge[bend left=10,above left]
          node[xshift=20pt,yshift=4pt]{\trFe{\lock[]\land\ul{\fm}}\Label{u_2}{join}{s}} 
          (011.west)
    (000.west) edge[bend right=10,above left]
          node{\trFe{\lock[]\land\ul{\fm}}\\\Label{u_1}{join}{s}} 
          (101.north)
    (000) edge[bend left=10,above, sloped] 
          node{\trFe{\unlock[]\land\ul{\fm}}\\\Label{u_1}{join}{s}}
          (200)
    (000) edge[bend left=5, above, sloped]
          node[xshift=6pt]{\trFe{\unlock[]\land \ul{\unlock[]\land\lnot\,\lock[]}}\\\Label{u_1,u_2}{join}{s}} 
          (220)
    (011.east) edge[bend left=10,above right]
          node[xshift=-20pt,yshift=4pt]{\trFe{\lock[]\land\ul{\fm}}\Label{s}{confirm}{u_2}} 
          (020.north)
    (020) edge[bend left=10,above] 
          node{\trFe{\top \land \ul{\fm}}\Label{u_2}{leave}{s}} 
          (000)
    (020.east) edge[bend left=10, above right] 
          node[yshift=5pt,xshift=-5pt]{\trFe{\lock[]\land\ul{\fm}}\\\Label{u_1}{join}{s}} 
          (121.north)
    (020) edge[bend left=10,above,sloped] 
          node{\trFe{\unlock[]\land \underline{\fm}}\\\Label{u_1}{join}{s}} 
          (220)
    (121.south) edge[bend left=10, right] 
          node[yshift=-8pt,xshift=-8pt]{\trFe{\lock[]\land\ul{\fm}}\\\Label{s}{confirm}{u_1}} 
          (220.east)
    (220) edge[bend left=12,below] 
          node{\trFe{\top \land \ul{\fm}}\Label{u_2}{leave}{s}} 
          (200)
    (200) edge[bend left=5,below] 
          node{\trFe{\unlock[]\land \underline{\fm}}\Label{u_2}{join}{s}} 
          (220)
    (200) edge[bend left=10,above,sloped] 
          node{\trFe{\top \land \ul{\fm}}\\\Label{u_1}{leave}{s}} 
          (000)
    (200.south) edge[bend right=10,below left]
          node[yshift=4pt]{\trFe{\lock[]\land\ul{\fm}}\\\Label{u_2}{join}{s}} 
          (211.west)
    (220) edge[bend left=5,below, sloped] 
          node[xshift=-6pt]{\trFe{\top \land \ul{\unlock[]\land\lnot\,\lock[]}}\\\Label{u_1,u_2}{leave}{s}} 
          (000)
    (220) edge[bend left=10,above,sloped] 
          node{\trFe{\top \land \ul{\fm}}\\\Label{u_1}{leave}{s}} 
          (020)
    (101.south) edge[bend right=10, left] 
          node[yshift=-8pt,xshift=8pt]{\trFe{\lock[]\land\ul{\fm}}\\\Label{s}{confirm}{u_1}} 
          (200.west)
    (211.east) edge[bend right=10,below right] 
          node[xshift=-5pt,yshift=4pt]{\trFe{\lock[]\land\ul{\fm}}\\\Label{s}{confirm}{u_2}} 
          (220.south);
  \coordinate[above=0.3 of 000,xshift=-3mm](initial);
  \draw[->,thick] (initial) -> (000);
  \node[left=4 of 011] (fm) {\shl{$\fm = \Fe{\lock[\textcolor{black!80}]\xor \unlock[\textcolor{black!80}]}$}};
\end{tikzpicture}
\vspace*{-\baselineskip}
\caption{
Generated \FETA $\fstype_\auth[\S_\auth]$ 
}
\label{fig:feta}
\end{figure}%

\subsection{\FETA versus \ETA}

\FETA are not simple extensions of
extended team automata (\ETA) introduced in~\cite{BHK20b}.
An \ETA is an LTS $\stype[\S]$ generated over a system \S of CA by 
an \STS \stype that explicitly filters the system transitions that satisfy the
synchronisation types determined by $\stype$. Concretely, an \ETA $\stype[\S]$ is the LTS $(Q,I,\Sigma,\stype[E])$, where
$Q,I,\Sigma$, and $E$ are induced by $\S$, and
$\stype[E] = \set{ t \in E \mid 
    t \models \stype(t.a)
   }$.


Observe that an \STS thus restricts the set of system transitions of a system $\S$, 
such that the \ETA $\stype[\S]$ has only a subset of the transitions of $\S$.
Instead,
an \FSTS 
and the local transition constraints of the components
$\A_i$ impose
transition constraints $\gamma$
on the system transitions of
an \FSys $\S$ such that
the \FETA $\fstype[\S]$ has all transitions of $\S$, but
appropriately constrained such that many of them will not be realisable anymore
for concrete products.


\smallskip
The next theorem shows that, for any valid product $p$,
the projection onto $p$ of the \FETA $\fstype[\S]$, generated over the \FSys 
$\S$ by the \FSTS $\fstype$, 
is the same as the \ETA over the projected system $\S^p$
generated by the projected \STS $\fstype^p$.
This result justifies the soundness of the definition
of a generated \FETA, in particular of its transition constraint $\gamma$.
It also shows that the diagram in \cref{fig:overview} commutes.


\begin{restatable}{theorem}{thmcommuting}\label{thm:commuting}
  Let \S be an \FSys 
  with feature model \fm,
  let \fstype be an \FSTS, and let $p\in \sem{\fm}$ be a valid product.
  Then: 
  \begin{align*}
     \fstype[\S]^p = \fstype^p[\S^p]. 
   \end{align*} 
\end{restatable}



\section{Receptiveness}
\label{sec:receptiveness}

As explained in \cref{{sec:featured_team_automata}} 
and formalised in~\cref{thm:commuting},
a \FETA $\fstype[\S]$ can be projected to a 
product $p\!\in\!\sem{\fm}$,
thus yielding an \ETA (i.e.\ a team) $\fstype[\S]^p\!=\!\fstype^p[\S^p]$. 
Any such \ETA describes the behaviour of a concrete system $\S^p$
whose components (the team members)
are coordinated by 
the synchronisation type specification $\stype\!=\!\fstype^p$.
This section analyses
com\-mu\-ni\-ca\-tion-safety of such families of \ETA.
Our 
aim is to provide criteria
on the level of \FETA that guarantee 
com\-mu\-ni\-ca\-tion-safety properties
for all \ETA obtained by 
projection (cf.~\Cref{sec:FETA-to-ETA}).

\subsection{Receptiveness for \ETA}\label{sec:receptiveness-ETA}

We focus on the property of receptiveness, which 
has been studied before
in the literature~\cite{AH01,CC02,LNW07}, mainly in the context of peer-to-peer communication.
An extension to multi-component communications was studied in~\cite{CK13} and in~\cite{BHK20b}, where also a notion of responsiveness not considered here was introduced.
The idea of receptiveness is as follows:
whenever, in a reachable state $q$ of an \ETA $\stype[\S]$, a group of components $J$ is (locally) enabled to perform an 
output action~$a$ 
such that its synchronous execution is in accordance
with the synchronisation type $\stype(a)$, 
we get a receptiveness requirement, written as $\rcp(J,a)@q$.
The \ETA is \emph{compliant} with this requirement if 
$J$ can find partners in the team which synchronise with the components in~$J$
by taking (receiving) $a$ as input.
If reception is immediate, we talk about receptiveness;
if the other components first perform some intermediate actions before accepting $a$, we talk about weak receptiveness.

Formally, receptiveness requirements, compliance, and receptiveness are defined as follows and illustrated in~\cref{ex:receptiveness}. We assume a given \ETA $\stype[\S] = (Q,\ab I,\ab \Sigma,\stype[E])$ generated by the \STS $\stype$ over a system $\S = (\N,(A_i)_{i\in \N})$ 
of \CA~$A_i$.

A \emph{receptiveness requirement} (\REQ) 
is an expression $\rcp(J,a)@q$,
where $q\in Q$ is a reachable state of $\stype[\S]$, $a \in \com$ is an 
action, and $\emptyset \neq J \subseteq \N$ is a set of component names such that $\forall_{j\in J}\cdot a \in \Sigma^!_j \land a~\bkw{en}_{A_j}@q_j$
and
$\stype(a)=(s,r) \Rightarrow |J|\in s \land 0\notin r$.
%
The last condition requires that i)~the number of components in $J$ fits the number of allowed senders according to the synchronisation type of $a$,
and ii)~at least one receiver must exist according to the synchronisation type of $a$.\footnote{Otherwise, the components in $J$ could simply output $a$ without reception.}
Hence our subsequent compliance and receptiveness notions, taken
from~\cite{BHK20b} and formalising the informal explanations above, depend strongly
on the synchronisation types of actions. 

\smallskip

The \ETA $\stype[\S]$ is \emph{compliant} with a \REQ
$\rcp(J,a)@q$ if the following holds:
\begin{align*}
& \exists_{R\neq\emptyset \text{ and } q' \in Q} \cdot q\tr{(J,a,R)}_{\stype[\S]} q'.
%
\intertext{\indent The \ETA $\stype[\S]$ is \emph{weakly compliant} with 
a \REQ $\rcp(J,a)@q$ if}
& \exists_{R\neq\emptyset \text{ and } \hat{q}, q' \in Q} \cdot 
		q \tr{\Lambda_{\setminus J}}\!{^*}_{\!\!\!\stype[\S]} \,
        \hat{q}\tr{(J,a,R)}_{\stype[\S]} q',
\end{align*}
where $\Lambda_{\setminus J}$ denotes the set of system labels in which no
component of $J$ participates. Indeed, only when state $\hat{q}$ is reached, the components of $J$ can actively get rid of their output.

\smallskip

The \ETA $\stype[\S]$ is \emph{(weakly) receptive} if it is (weakly) compliant with
\emph{all} \REQ{s} 
for $\stype[\S]$.

\begin{example}\label{ex:receptiveness}
%
Let \ETA $\fstype^{\Fe{\lock}}[{\S_\auth^{\Fe{\lock}}}]$ be generated by the \STS $\fstype^{\Fe{\lock}}$  (i.e.\ the projection 
of $\fstype$ from \cref{ex:feta} to $\{\Fe{\lock}\}$)
over the system ${\S_\auth^{\Fe{\lock}}} = 
(\mc{N},\set{A_{u_1}^{\Fe{\lock}},A_{u_2}^{\Fe{\lock}},A_{s}^{\Fe{\lock}}})$ of \cref{ex:fsys},
with $\fstype^{\Fe{\lock}}(\mi{join}) = \fstype^{\Fe{\lock}}(\mi{confirm}) =
\fstype^{\Fe{\lock}}(\mi{leave}) = ([1,1],[1,1])$.
%
In the initial global state $(0,0,0)$ both users are enabled to execute output action \mi{join}, but not simultaneously.
%
Hence, we get two \REQs
$\rcp(\{u_i\},\mi{join})@(0,0,0)$, one for each $i\in\{1,2\}$.
The \ETA $\fstype^{\Fe{\lock}}[{\S_\auth^{\Fe{\lock}}}]$ is compliant
with both \REQs because
$(0,0,0)\tr{(\{u_1\},\mi{join},\{s\})}_{\fstype^{\Fe{\lock}}[{\S_\auth^{\Fe{\lock}}}]} (1,0,1)$
and $(0,0,0)\tr{(\{u_2\},\mi{join},\{s\})}_{\fstype^{\Fe{\lock}}[{\S_\auth^{\Fe{\lock}}}]} (0,1,1)$.
%
Now assume that user $A_{u_1}^{\Fe{\lock}}$ joins. Then $\fstype^{\Fe{\lock}}[{\S_\auth^{\Fe{\lock}}}]$ ends up in state $(1,0,1)$,
where user $A_{u_2}^{\Fe{\lock}}$ may decide to join, i.e.\ there is
a \REQ $\rcp(\{u_2\},\mi{join})@(1,0,1)$.
But 
the server is not yet ready for $A_{u_2}^{\Fe{\lock}}$ as it first needs to send a confirmation to $A_{u_1}^{\Fe{\lock}}$. Therefore $\fstype^{\Fe{\lock}}[{\S_\auth^{\Fe{\lock}}}]$ is not compliant with $\rcp(\{u_2\},\mi{join})@(1,0,1)$, but it is weakly compliant with this \REQ.
We can show that the \ETA $\fstype^{\Fe{\lock}}[\S_\auth^{\Fe{\lock}}]$
is either compliant or weakly compliant with any \REQ and therefore it is weakly receptive.

Next, consider \ETA $\fstype^{\Fe{\unlock}}[{\S_\auth^{\Fe{\unlock}}}]$ generated by the \STS $\fstype^{\Fe{\unlock}}$ over the system ${\S_\auth^{\Fe{\unlock}}} = 
(\mc{N}, \set{A_{u_1}^{\Fe{\unlock}},A_{u_2}^{\Fe{\unlock}},A_{s}^{\Fe{\unlock}}})$ of \cref{ex:fsys}
with $\fstype^{\Fe{\unlock}}(\mi{join}) = \fstype^{\Fe{\unlock}}(\mi{leave}) = ([1,*],[1,1])$.
%
%
In state $(0,0,0)$, both users are enabled to output \mi{join}. Therefore, according to the sending multiplicity
$[1,*]$ of $\fstype^{\Fe{\unlock}}(\mi{join})$, there are three \REQs for that state, among which
$\rcp(\{u_1,u_2\},\mi{join})@(0,0,0)$. 
Note that $\fstype^{\Fe{\unlock}}[{\S_\auth^{\Fe{\unlock}}}]$ is compliant with this \REQ due to the team transition
$(0,0,0)\tr{(\{u_1,u_2\},\mi{join},\{s\})}_{\fstype^{\Fe{\unlock}}[{\S_\auth^{\Fe{\unlock}}}]} (1,1,1)$.
In fact, the \ETA $\fstype^{\Fe{\unlock}}[\S_\auth^{\Fe{\unlock}}]$
is compliant with all \REQs and therefore it is receptive.
\qedx
\end{example}

\subsection{Featured Receptiveness for \FETA}\label{sec:receptiveness-FETA}

We now turn to \FETA and discuss how the notions of receptiveness requirements, compliance, and receptiveness can be transferred to the feature level.
We assume a given \FETA $\fstype[\S] = (Q,\ab I,\ab \Sigma,\ab E,\ab \mshl{F},\ab \mshl{\fm},\ab \mshl{\gamma})$ generated by the \FSTS $\fstype$ over an \FSys $\S = (\N,\ab (A_i)_{i\in \N})$, 
with \FCA $A_i$. 
The crucial difference with the case of \ETA is that \FETA 
are based on syntactic specifications modelling families of teams.
Hence a \REQ 
$\rcp(J,a)@q$ formulated for an \ETA
cannot be formulated for a \FETA as it is.
Instead, it must take into account the valid products $p$ of the family for which the requirement is meaningful.
%
For this purpose, we propose to complement
$\rcp(J,a)@q$ by a syntactic application condition, resulting in a \emph{featured receptiveness requirement} (\FREQ), written as
$\trFe{\bkw{prod}(J,a,q)}\,\rcp(J,a)@q$.
Herein ${\bkw{prod}(J,a,q)}$ is a feature expression, which
characterises the set of valid products for which the \REQ
$\rcp(J,a)@q$
is applicable for $\fstype[\S]^p$.
The expression 
${\bkw{prod}(J,a,q)}= {\bkw{fe}(J,a,q)} \land \Chi{P(\fstype,J,a)}
\land \Chi{P(q)}$
consists of the following parts:

%
%
\begin{enumerate}
    \item ${\bkw{fe}(J,a,q)}$ = $\bigwedge_{j\in J} \bigvee \gamma_j(q_j\tr{a}_{A_j} q_j')$ combines the feature expressions of all transitions of components
    $\A_j$ ($j \in J$) with action $a$ and starting in the local state $q_j$.
    For any \FCA $\A_j$, the disjunction
    $\bigvee \gamma_j(q_j\tr{a}_{A_j} q_j')$
    ranges over the feature expressions of all local transitions of $A_j$
    starting in $q_j$ and labelled with $a$.
    Hence, if there are more such transitions it is sufficient if one of them is realised (in a projection of $\A_j$).
    Thus ${\bkw{fe}(J,a,q)}$ characterises those products~$p$ for which outgoing transitions with output $a$ are realisable in the local states $q_j$
    of $A_j$
    and hence enabled in $q_j$ in the projected component $\A_j^p$.
    \item
    $\Chi{P(\fstype,J,a)}$ is the feature expression which characterises
    (cf.~\cref{sec:preliminaries}) the set $P(\fstype,J,a) = \{p \in\sem{\fm} \mid \fstype(p,a){=}(s,r) \Rightarrow |J|\in s \land 0\notin r\}$.
    This is the set of all products $p$ such that $\fstype(p,a)$
    allows $|J|$ as number of senders and requires at least one receiver.
    \item
    $\Chi{P(q)}$ is the feature expression which characterises the set $P(q)$ of products for which state $q$ is reachable by transitions of $\fstype[\S]$ whose constraints are satisfied by $p$, i.e.\
    $P(q) = \{p \in\sem{\fm} \mid \exists_{q_0 \in I} \cdot~ 
q_0\tr{l_1}_{\fstype[\S]}q_1\tr{l_2}\dots\tr{l_n}_{\fstype[\S]}q_n = q
\text{ for some } n\geq0, \text{ and } p \models \gamma(q_{i-1}\tr{l_i}_{\fstype[\S]}q_i) \text{ for } i = 1,\ldots,n \}.$
\end{enumerate}

In summary, an \FREQ for $\fstype[\S]$
has the form
$\trFe{\bkw{prod}(J,a,q)}\,\ab\rcp(J,a)@q$,
where $q\in Q$ is a reachable state of $\fstype[\S]$, $a \in \com$, $\emptyset \neq J \subseteq \N$ is a set of component names
such that $\forall_{j\in J}\cdot a \in \Sigma^!_j \land a~\bkw{en}_{A_j}@q_j$,
and $\bkw{prod}(J,a,q)$ is a satisfiable feature expression as defined above. 
Note that
$\models {\bkw{prod}(J,a,q)} \rightarrow \fm$,
because $P(\fstype,J,a)$ in item~2 (and also $P(q)$ in item~3) is a subset of $\sem{\fm}$.

\smallskip

The following lemma provides a formal relation between
\REQs and \FREQs.



\begin{restatable}{lemma}{lemfeaturedreq}\label{lem:featured-req}
For all products $p$ it holds:
$\trFe{\bkw{prod}(J,a,q)}\,\ab\rcp(J,a)@q$ is an \FREQ for
$\fstype[\S]$ and $p \models \bkw{prod}(J,a,q)$
iff
$p \in \sem{\fm}$ and
$\rcp(J,a)@q$ is a \REQ for $\fstype[\S]^p$.
\end{restatable}




\begin{example}\label{ex:freq}
\cref{fig:feta-with-req} shows an excerpt of the \FETA $\fstype[\S_{\auth}]$ in \cref{fig:feta} depicting the \FREQ{s} for states $(0,0,0)$, $(0,1,1)$, and $(0,2,0)$. 
First note that an output of $\mi{join}$ is enabled at local state 0 in both components
$A_{u_1}$ and $A_{u_2}$. For $\rcp(\set{u_1},\mi{join})$ at state $(0,0,0)$ we get
$\bkw{fe}(\set{u_1},\mi{join}, (0,0,0)) = \lock\lor\unlock$
according to the constraints of both $\mi{join}$ transitions in $A_{u_1}$. 
Moreover, $P(\fstype,\set{u_1},\mi{join}) = \set{\set{\Fe{\lock}},\set{\Fe{\unlock}}} = \sem{\fm}$ and therefore $\Chi{P(\fstype,\set{u_1},\mi{join})}$ is equivalent to $\fm$. Also $P(0,0,0) = \set{\set{\Fe{\lock}},\set{\Fe{\unlock}}}$
since state $(0,0,0)$ is reachable in both products.
So $\bkw{prod}(\set{u_1},\mi{join}, (0,0,0)) =  (\Fe{\lock}\vee\Fe{\unlock})\land\fm \land\fm$, which reduces to
$\fm = \Fe{\lock}\xor\Fe{\unlock}$.
Thus we get the \FREQ $\trFe{\lock[]\xor\unlock[]}\,\rcp(\set{u_1},\mi{join})$ at $(0,0,0)$.
The case of $\{u_2\}$ is analogous.

Considering a possible simultaneous output of
$\mi{join}$ by $u_1$ and $u_2$ we get
$\bkw{fe}(\{u_1,u_2\},\ab\mi{join},\ab
(0,0,0)) =\ab  (\Fe{\lock}\vee\Fe{\unlock})
\vee (\Fe{\lock}\vee\Fe{\unlock}).$ 
And we get
$P(\fstype,\set{u_1,u_2},\mi{join}) = \set{\set{\Fe{\unlock}}}$, since only for
the product $\set{\Fe{\unlock}}$ a synchronisation of several users is allowed.
Therefore $\Chi{P(\fstype,\set{u_1,u_2},\mi{join})} = \Fe{\unlock}\land\neg\,\Fe{\lock}$.
As above, $\Chi{P(0,0,0)} = \fm$. 
Thus $\bkw{prod}(\set{u_1,u_2},\mi{join}, (0,0,0)) =  (\Fe{\lock}\vee\Fe{\unlock})\land (\Fe{\unlock}\land\neg\,\Fe{\lock})\land \fm$, which reduces to $\Fe{\unlock}\land\neg\,\Fe{\lock}$.
Hence we get the \FREQ $\trFe{\unlock[]\land\neg\,\lock[]}\,\rcp(\set{u_1,u_2},\mi{join})$ at $(0,0,0)$.

An interesting case is 
$\trFe{\lock[]\land\neg\,\unlock[]}\,\rcp(\set{u_1},\mi{join})$ at $(0,1,1)$.
Here $\bkw{fe}(\{u_1\},\ab\mi{join},\ab
(0,1,1))$ is again $\Fe{\lock}\vee\Fe{\unlock}$ and 
$\Chi{P(\fstype,\set{u_1},\mi{join})}$ is equivalent to $\fm$. 
However, the state $(0,1,1)$ is only reachable in the product $\set{\Fe{\lock}}$, i.e.\
$P(0,1,1) = \set{\set{\Fe{\lock}}}$.
Therefore, $\Chi{P(0,1,1)} = \Fe{\lock}\land\neg\,\Fe{\unlock}$.
In summary, 
$\bkw{prod}(\set{u_1,u_2},\mi{join}, (0,1,1)) =
(\lock\land\unlock)\land\ab \fm \land\ab (\Fe{\lock}\land\ab\neg\,\Fe{\unlock})$,
which reduces to $(\Fe{\lock}\land\ab\neg\,\Fe{\unlock})$. 
The other \FREQs are computed similarly.
\qedx
\end{example}

\begin{figure}[htb!]
  \centering

\begin{tikzpicture}[align=center, every node/.style={scale=.9,font=\footnotesize}]
  \node[rloc](000){$0,0,0$};
  \node[rloc,right=6 of 000](020){$0,2,0$};
  \coordinate(mid1) at ($(000)!.5!(020)$);
  \node[rloc,above=.8 of mid1](011){$0,1,1$};


  \node[left=1 of 000](200){$\cdots$};
  \node[right=1 of 020](220){$\cdots$};

  \path[->] 
    (000) edge[bend left=7,above] 
          node{\trFe{\unlock[]\land\ul{\fm}}\Label{u_2}{join}{s}} 
          (020)
          edge[bend left=10,above]
          node[yshift=7pt,xshift=6pt]{\trFe{\lock[]\land\ul{\fm}}\\\Label{u_2}{join}{s}} 
          (011)
          edge[bend left=0,dashed,<->] 
          (200)
    (011) edge[bend left=10,above]
          node[yshift=7pt,xshift=-6pt]{\trFe{\lock[]\land\ul{\fm}}\\\Label{s}{confirm}{u_2}} 
          (020)
    (020) edge[bend left=10,above] 
          node{\trFe{\top \land \ul{\fm}}\Label{u_2}{leave}{s}} 
          (000)
          edge[bend left=0,dashed,<->] 
          (220)
   ;
  \coordinate[above=0.3 of 000](initial);
  \draw[->,thick] (initial) -> (000);
    \node[above=1.9 of 020] (fm) {\qquad\qquad\qquad\shl{$\fm = \Fe{\lock[\textcolor{black!80}]\xor \unlock[\textcolor{black!80}]}$}};
    \node[req,above=0.8 of 000,xshift=-1cm](r000){
      \Req{
          \trFe{\lock[]\xor \unlock[]}~~\,\, \rcp(\set{u_1},\mi{join})\\
          \trFe{\lock[]\xor \unlock[]}~~\,\, \rcp(\set{u_2},\mi{join})\\
          \trFe{\unlock[]\land \neg\,\lock[]}~ \rcp(\{u_1,u_2\},\mi{join})  }
       {(s,\mi{join})_\Fe{\lock\xor \unlock}\lor(s,leave)_\Fe{\lock\xor \unlock}} 
    }; 
    \node[req,above=0.8 of 020,xshift=1cm](r020){
      \Req{
        \trFe{\lock[]\xor \unlock[]}~ \rcp(\set{u_1},\mi{join})\\
        \trFe{\lock[]\xor \unlock[]}~ \rcp(\set{u_2},\mi{leave})
      }{(s,\mi{join})_ \Fe{\lock[]\xor \unlock[]}\lor(s,\mi{leave})_ \Fe{\lock[]\xor \unlock[]}}
    };
    \node[req,above=0.5 of 011](r011){
      \Req{
        \trFe{\lock[]\land\neg\,\unlock[]}\,\, \rcp(\set{u_1},\mi{join})\\
        \trFe{\lock[]\land\neg\,\unlock[]}~ \rcp(\set{s},\mi{confirm})}{}
    };
    \path[-,dotted] 
    (000) edge[bend left=10] (r000)
    (020) edge[bend right=10] (r020)
    (011) edge[bend right=0] (r011);
\end{tikzpicture}
  \vspace*{-\baselineskip}
  \caption{Part of $\fstype[\S_\auth]$ from \cref{fig:feta} enriched with \REQs
  }
  \label{fig:feta-with-req}
\end{figure}


Next, we define featured compliance with an \FREQ. We use a logical formulation which, as we shall see, captures compliance for the whole family
of products.



\smallskip

The \FETA $\fstype[\S]$ is \emph{featured compliant} with an \FREQ
$\trFe{\psi}\,\rcp(J,a)@q$
if for some $n \geq 1$ and for $k=1,\ldots,n$ there exist
transitions
$t^k = q\tr{(J,a,R^k)}_{\fstype[\S]} q^k$ with $R^k\neq\emptyset$
such that
$\models \psi \rightarrow \bigvee_{k\in \{1,\ldots,n\}} \gamma(t^k).$

\smallskip

The definition of featured compliance can be unfolded by considering all $p \in \sem{\fm}$. This shows the relationship to the compliance notion for \ETA. 


\begin{restatable}{lemma}{lemfeaturedcompliance}\label{lem:featured-compliance}
Let $\trFe{\psi}\,\rcp(J,a)@q$ be an \FREQ for the \FETA $\fstype[\S]$. Then:
\\[-2mm]
$$
\left[\mwrap{\vspace{8mm}}\mwrap{\fstype[\S] \text{ is featured compliant}\\[5pt]
 \text{with }\trFe{\psi}\,\rcp(J,a)@q}
\right]
\!\mwrap{\bm\Leftrightarrow}\!
\left[\mwrap{
  \forall_{p \subseteq F \text{ with } p \models {\psi}}\cdot
  ~~
  \exists_{R\neq\emptyset \text{ and } q' \in Q} \cdot
  \\[5pt]
  q\tr{(J,a,R)}_{\fstype[\S]}q' \text{ and } p \models \gamma(q\tr{(J,a,R)}_{\fstype[\S]}q')
}
\right]
$$
\end{restatable}



The next definition generalises featured compliance to featured weak compliance.
It is a technical but straightforward extension that transfers the concept of weak receptiveness to the featured level.

\smallskip
The \FETA $\fstype[\S]$ is \emph{featured weakly compliant} with an \FREQ
$\trFe{\psi}\,\rcp(J,a)@q$
if for some $n \geq 1$ and for $k=1,\ldots,n$ there exist
sequences $\sigma^k$ of transitions
$$\sigma^k = q^k_0 \tr{(S^k_0,a^k_0,R^k_0)}_{\fstype[\S]} q^k_1\,\, \cdots \,\,
q^k_{m_k} \tr{(S^k_{{m_k}},a,R^k_{{m_k}})}_{\fstype[\S]} q^k_{m_{k}+1}$$
with 
$q^k_0 = q, m_k \geq 0$,
$(S^k_i\cup R^k_i) \cap J = \emptyset$ for $i= 0,\ldots, m_{k-1}$,
$R^k_i\neq\emptyset$ for $i=0,\ldots, m_k$, and $S^k_{{m_k}} = J$
such that 
  $$\models \psi \rightarrow \bigvee_{k\in \{1,\ldots,n\}}
\bigwedge_{i\in \{0,\ldots,m_k\}}\gamma(q^k_i
\tr{(S^k_i,a^k_i,R^k_i)}_{\fstype[\S]} q^k_{i+1}).$$

%
We remark that~\cref{lem:featured-compliance} can be extended in a straightforward way
to characterise featured weak compliance.

\smallskip

The \FETA $\fstype[\S]$ is \emph{featured (weakly) receptive} if it is  featured (weakly) compliant with all
\FREQ{s} for $\fstype[\S]$.

\begin{example}\label{ex:freceptive}
We consider some \FREQs for the \FETA $\fstype[\S_\auth]$
as depicted in~\cref{fig:feta-with-req}. 
The first \FREQ is
$\trFe{\lock[]\xor\unlock[]}\,\rcp(\set{u_1},\mi{join})\ab@(0,0,0)$.
As we can see in~\cref{fig:feta}, there are two transitions, say $t_1, t_2$, in $\fstype[\S_@]$ with source state $(0,0,0)$ and label
$(\set{u_1},\mi{join},\set{s})$,
such that $\gamma(t_1)=\lock\land\fm$ and $\gamma(t_2)=\unlock\land\fm.$  
Hence, for checking featured compliance with this \FREQ we have to prove:
$$\models \Fe{\lock\xor\unlock} \rightarrow \Fe{(\lock\land\fm)\vee(\unlock\land\fm)}.$$
But this is easy, since the conclusion
is equivalent to $\fm=\Fe{\lock\xor\unlock}$. To achieve this it is essential
to have the disjunction of $\gamma(t_1)$ and $\gamma(t_2)$ in the conclusion.

As a second \FREQ we consider
$\trFe{\unlock[]\land\neg\,\lock[]}\,\rcp(\set{u_1,u_2},\mi{join})@(0,0,0)$.
As we can see in~\cref{fig:feta}, there is one transition in $\fstype[\S_@]$ with source state $(0,0,0)$ and label
$(\set{u_1,u_2},\mi{join},\set{s})$,
which has the transition constraint $\Fe{\unlock\land\neg\,\lock}$.
Featured compliance with this \FREQ holds trivially, since
$$\models \Fe{\unlock}\land\neg\,\Fe{\lock} \rightarrow \Fe{\unlock\land\neg\,\lock}.$$
As a last example, consider the \FREQ
$\trFe{\lock[]\land\neg\,\unlock[]}\,\rcp(\set{u_1},\mi{join})@(0,1,1)$.
In state $(0,1,1)$, no transition with action $\mi{join}$ can be performed
by the \FETA $\fstype[\S_\auth]$. Therefore featured compliance
does not hold. However, featured weak compliance holds for the following reasons.
We take $n=1$ (in the definition of featured weak compliance)
and select, in~\cref{fig:feta}, the transition sequence
\begin{align*}
(0,1,1) \tr{\trFe{\lock[] \land \underline{\fm}}\mathit{(\{s\},\,confirm,\{u_2\})}} (0,2,0)
\tr{\trFe{\lock[] \land \underline{\fm}}(\{u_1\},\,\mathit{join},\{s\})} (1,2,1).
\end{align*}
Then, we get the following proof obligation (conjoining the constraints of the two consecutive transitions in the conclusion):
$\models (\Fe{\lock}\land\neg\,\Fe{\unlock}) \rightarrow
\Fe{(\lock \land \fm) \land (\lock \land \fm)}.$
Obviously, this holds since the conclusion reduces to $\Fe{\lock}\land\neg\,\Fe{\unlock}$.
We can show that the \FETA $\fstype[\S_\auth]$
is either featured compliant or featured weakly compliant with any \FREQ and therefore it is featured weakly receptive.
\qedx
\end{example}

\subsection{From Featured Receptiveness to Receptiveness}
\label{sec:FETA-to-ETA}

This section presents our main result.
We show that instead of checking product-wise each member of a family
of product configurations for (weak) receptiveness,
it is sufficient to verify once
featured (weak) receptiveness for the family model.
We can even show that this technique is not only sound but also complete
in the sense, that if we disprove featured (weak) receptiveness
on the family level, then there will be a product for which the
projection is not (weakly) receptive.



\begin{restatable}
{theorem}{thmfrectorec}
\label{thm:frec-to-rec}
Let $S$ be an \FSys 
with feature model $\fm$, let $\fstype$ be an \FSTS, and let $\fstype[S]$ be its generated \FETA. Then:
\begin{align*}
  \Big[\fstype[\S]\text{ is featured (weakly) receptive} \Big]
  \bm{\Leftrightarrow}
  \Big[\forall_{p \in \sem{\fm}} \cdot \fstype[\S]^p\text{ is (weakly) receptive}\Big]
\end{align*}
\end{restatable}

\begin{example}
In~\cref{ex:freceptive} we showed that the \FETA $\fstype[\S_\auth]$
is featured weakly receptive.
Therefore, by applying~\cref{thm:frec-to-rec}, we know
that for both products $\{\Fe{\lock}\}$ and $\{\Fe{\unlock}\}$, 
the \ETA $\fstype^{\Fe{\lock}}[{\S_\auth^{\Fe{\lock}}}]$
and $\fstype^{\Fe{\unlock}}[{\S_\auth^{\Fe{\unlock}}}]$
are weakly receptive (a result which we checked product-wise
in~\cref{ex:receptiveness}).
\qedx
\end{example}

\paragraph{\rm\bf Note on Complexity}


Note that an \FREQ for $\fstype[\S]$ necessarily involves a syntactic application condition, which is a feature expression that characterises the set of valid products $p$ for which the featureless \REQ is applicable for $\fstype[\S]^p$. Part of this feature expression is a characterisation $\Chi{P(q)}$ of the set $P(q)$ of products for which state $q$ is reachable by transitions of $\fstype[\S]$ whose constraints are satisfied by $p$, 
which requires a reachability check for $q$.
This may seem computationally expensive. However, it has been shown that static analysis of properties of \FTSs that concern the reachability of states and transitions in valid products (LTSs) is feasible in reasonable time even for \FTSs of considerable size, by reducing the analysis to SAT solving~\cite{BDLMP21}. 
In fact, while SAT solving is NP-complete, SAT solvers are effectively used for static analysis of feature models with hundreds of thousands of clauses and tens of thousands of variables~\cite{MWC09,LGCR15}. 
Finally, we note that the results presented in this section are still sound, 
but not complete, without the aforementioned characterisation
of $P(q)$.



\section{Tool Support}
\label{sec:implementation}

\begin{figure}[htb!]
  \centering
  \begin{tabular}[b]{@{}l@{}}
    \numsshot[.484]{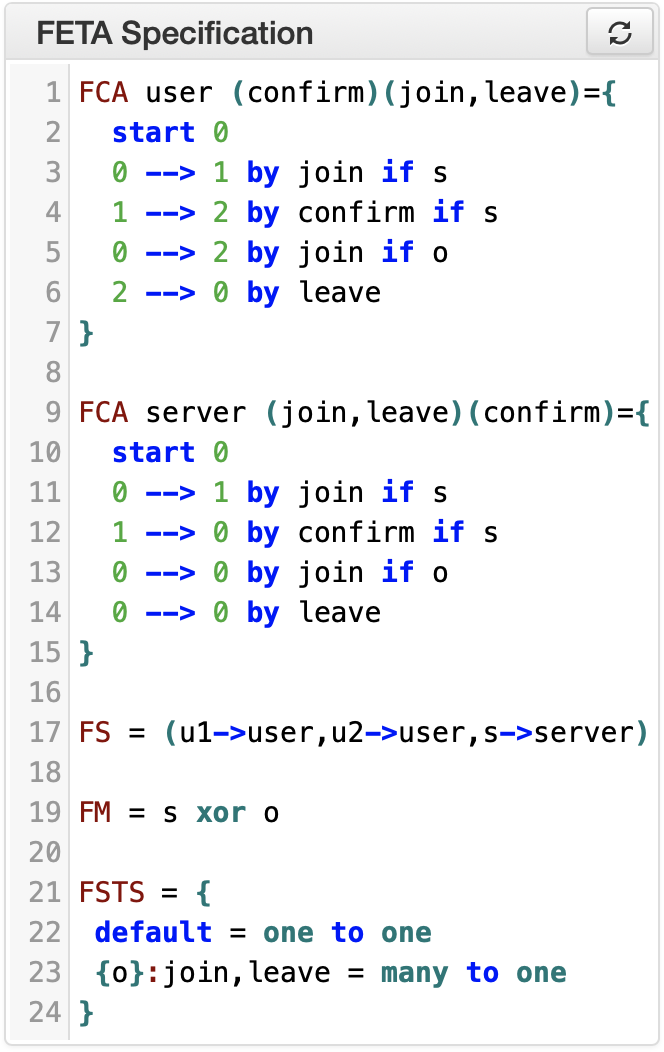}{1}
  \end{tabular}
  \begin{tabular}[b]{@{}l@{}}
    \numsshot[.47]{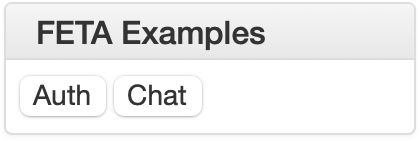}{3}\\[22pt]\numsshot[.47]{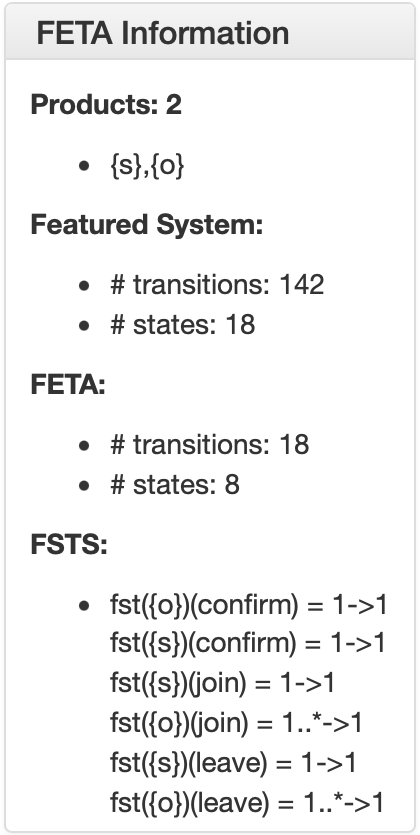}{5}
  \end{tabular}%
  \begin{tabular}[b]{@{}l@{}}
    \numsshot[.47]{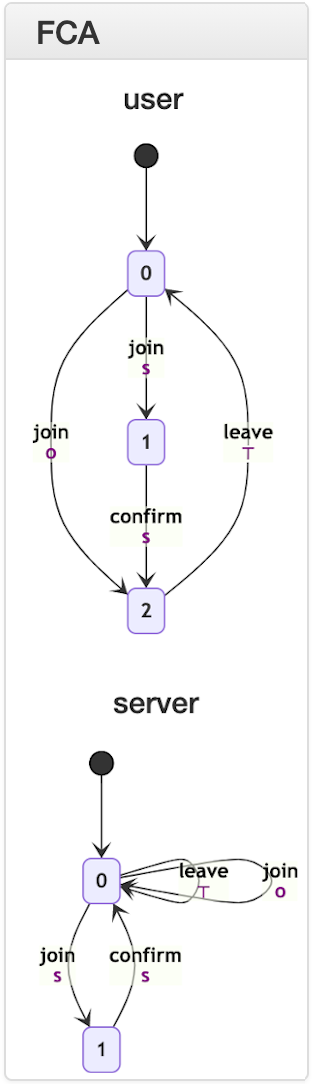}{4}
  \end{tabular}
  \\
  \begin{tabular}[b]{@{}l@{}}
    \numsshot[.37]{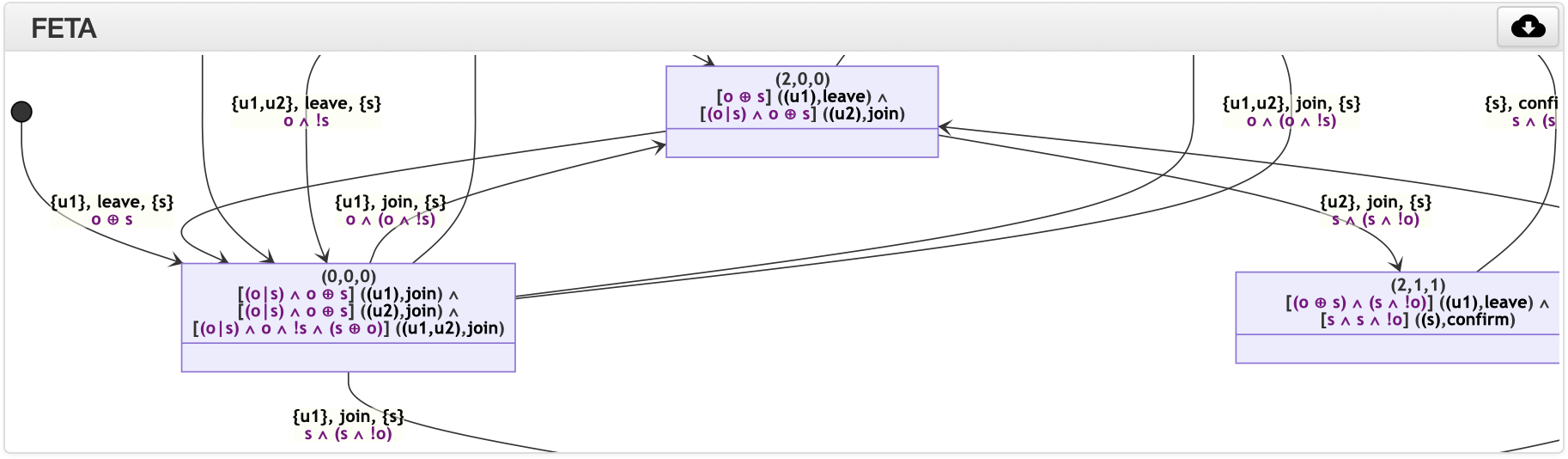}{2}
  \end{tabular}
  \caption{Screenshots of the widgets in the online tool for \FETA} 
  \label{fig:tool}
\end{figure}

We implemented a prototypical tool to specify and analyse 
\FETA. This requires to 
define an \FSys over a set of \FCA, a shared feature model, and an \FSTS.
The tool can be used online 
and downloaded at 
\url{https://github.com/arcalab/team-a}.
The interface is organised by 5 widgets 
(illustrated in \cref{fig:tool}):
{\small\cn{1}}~%
a text editor to specify a \FETA, using a dedicated domain-specific language; 
{\small\cn{2}}~%
an \FTS view of the \FETA, together with the \FREQs generated automatically for each state, similar to \cref{fig:feta-with-req}; 
{\small\cn{3}}~%
a set of example \FETA; 
{\small\cn{4}}~%
a view of each individual \FCA, similar to \cref{fig:fca-ex}; and 
{\small\cn{5}}~%
some statistics of the various models,  
 including the number of states, transitions, features, and products.

%
The tool is written in Scala and it uses the Play Framework
to generate an 
interactive website using a client-server architecture. The Scala 
code is compiled into JavaScript using \texttt{Scala.js}
to run on the client side, 
and into JVM binaries that run on the server side. 
The server side is currently needed to use an off-the-shelf Java library, \texttt{Sat4j}, 
to find all products that satisfy a feature model.

\section{Conclusion} 
\label{sec:conclusion_and_discussion}

We introduced featured team automata to specify and analyse systems of 
featured component automata
and to explore composition and communication-safety. We showed that family-based analysis of receptiveness (no message loss) suffices to study receptiveness of product configurations. We implemented our theory in a prototypical tool. 


In the future, we intend to extend our theory
 to address i)~responsiveness, i.e.\ no indefinite waiting for input,
 and ii)~compositionality, i.e.\ extend \FETA to composition of systems (that behaves well with \FSTSs) and investigate conditions under which communication safety is preserved by \FETA composition. Moreover, we will further develop the tool
 and analyse the practical impact of \FETA on the basis of larger case studies.
 This involves a thorough study of the efficiency of featured receptiveness checking compared to product-wise checking of receptiveness.
Finally, we aim to implement a family-based analysis algorithm that computes, for a given \FETA, the set of all product configurations that yield communication-safe systems.





\paragraph{\rm\bf Acknowledgments.}

\begin{small}
Ter Beek received funding from the MIUR PRIN 2017FTXR7S project IT MaTTerS (Methods and Tools for Trustworthy Smart Systems). 
Cledou and Proen\c{c}a received funding from the ERDF -- European Regional Development Fund through the Operational Programme for Competitiveness and Internationalisation -- COMPETE 2020 Programme (project DaVinci, POCI-01-0145-FEDER-029946) and by National Funds through the Portuguese funding agency, FCT -- Fundação para a Ciência e a Tecnologia.
Proen\c{c}a also received 
National Funds through FCT/MCTES, within the CISTER Research Unit (UID/CEC/04234);
by the Norte Portugal Regional Operational Programme -- NORTE 2020 (project REASSURE, NORTE-01-0145-FEDER-028550) under the Portugal 2020 Partnership Agreement, through ERDF the FCT;
and European Funds through the ECSEL Joint Undertaking (JU) under grant agreement No~876852 (project VALU3S).




\end{small}

\bibliographystyle{splncs04}
\bibliography{biblio}

\newpage
\appendix

\section{Proofs}
\label{sec:proofs}


\thmcommuting* 
\begin{proof}
  Both \LTSs share the same structure from \S: the state space~$Q$, the initial states~$I$, and the system labels.
  Furthermore, the transitions of both systems are a subset of the induced set of system transitions $E$ from \S.
  It remains to show that $E_1=E_2$, where
    $E_1\subseteq E$ is the set of transitions of the first LTS
    and
    $E_2\subseteq E$ the one of the second. 
  
  Let
    $\gamma$ be the transition constraint of $\fstype[\S]$ and,
    for each $i\in \N$ of \S, $\gamma_i$ be the local transition constraint
    of 
    $A_i$.
  Also recall that
  % ...
  $\hat{\gamma}(q\tr{(S,a,R)}q') = \bigwedge_{i\in (S\cup R)} \gamma_i(q_i\tr{(S,a,R)}_{A_i} q_i')$,
  $S = (\N,(A_i)_{i\in \N})$,
  and
  $S^p = (\N,(A^p_i)_{i\in \N})$.
  The set of transitions of \S and from $\S^p$ are, respectively:
  {\small
  \\--
  $E = \{q\tr{(S,a,R)}q' \mid 
    \forall_{i\in(S\cup R)} \cdot q_i\tr{a}_{A_i}\!q'_i, ~
    \forall_{j\in\N\setminus (S\cup R)} \cdot q_j{=}q_j'
    \}$
  \\--
  $E[S^p]\mkern-1mu=\mkern-1mu\{q\tr{(S,a,R)}q' \mid 
    \forall_{i\in(S\cup R)} \cdot [q_i\tr{a}_{A_i}\!q'_i\mkern2mu\land\mkern2mu
      p \models \gamma_i(q_i\tr{a}_{A_i} q'_i)],\,
    \forall_{j\in\N\setminus (S\cup R)} \cdot q_j{=}q_j'
    \}$
  }\\[.25em]
  Finally, we show that $E_1=E_2$. By definition, $E_2=\set{t\in E[\S^p]\mid t\models \fstype^p}$, and:
  \newcommand{\nxte}[1][=]{\\#1 &~}
  \newcommand{\nxt}[2][=]{\\\textcolor{gray}{\{#2\}}~#1&~}
  \newcommand{\nxtt}[1]{\nxt{\textit{#1}}}
  %
%
  \begin{align*}
    &~ E_1 
        \nxtt{by def.}
    \set{t\in E\mid p\models \gamma(t)}
        \nxte
    \set{ t\in E \mid p\models [\hat{\gamma}(t) \land \Chi{P(\fstype,t)}]]}
        \nxte
    \set{ t\in E \mid [p\models \hat{\gamma}(t)] \land
                     [p\models \Chi{P(\fstype,t)}]}
        \nxte
    \set{t\in E\mid [p\models \hat{\gamma}(t)] \land
                        [p\in \sem{\fm}] \land [t\models \fstype^p]}
        \nxt{p\in\sem\fm}
    \set{t\in E\mid [p\models \hat{\gamma}(t)] \land
                        [t\models \fstype^p]}
        \nxtt{unfolding $E$, $E[\S^p]$, $\hat\gamma$}
    E_2
   \qedhere\end{align*}
\end{proof}

\lemfeaturedreq*
\begin{proof}
($\bm\Rightarrow)$:
By definition of $\bkw{prod}(J,a,q)$
(in particular of $\Chi{P(\fstype,J,a)}$ and $\Chi{P(q)}$),
$\models {\bkw{prod}(J,a,q)} \rightarrow \fm$.
Hence, $p \models \bkw{prod}(J,a,q)$ implies $p \in \sem{\fm}$.  
By assumption, we know that $q$ is a reachable state of $\fstype[\S]$ and
$p\models \Chi{P(q)}$, i.e.\ $p \in P(q)$. 
Therefore, $q$ is a reachable state of $\fstype[\S]^p.$
Moreover, we know that
$\emptyset \neq J \subseteq \N$ and $\forall_{j\in J}\cdot a \in \Sigma^!_j \land\ab a~\bkw{en}_{A_j}@q_j$. Since, by assumption,
$p \models \bkw{fe}(J,a,q)$, there exists in each $A_j$ a transition $t_j$ starting in $q_j$ with action $a$ such that $p\models \gamma_j(t_j)$.
Hence, by definition of projection, $a~\bkw{en}_{A^p_j}@q_j$ for all $j \in J$.
Finally, we know that $p \models \Chi{P(\fstype,J,a)}$. 
Therefore, $\fstype(p,a)=(s,r) \Rightarrow |J|\in s \land 0\notin r$.
Since $\fstype^p(a) =\fstype(p,a)$ we get $\fstype^p(a)=(s,r) \Rightarrow |J|\in s \land 0\notin r$.
Thus, $\rcp(J,a)@q$ is a \REQ for the \ETA $\fstype^p[\S^p]$
which, by~\cref{thm:commuting}, coincides with $\fstype[\S]^p$.  

\smallskip

($\bm\Leftarrow$): Let $p\in\sem{\fm}$ and $\rcp(J,a)@q$ be a \REQ for $\fstype[\S]^p$
which is the \ETA $\fstype^p[\S^p]$.
Then $q$ is a reachable state of $\fstype[\S]^p$
and therefore, by definition of \FTS projection,
$q$ is reachable in $\fstype[\S]$ by transitions whose constraints are satisfied by $p$. Hence, $q$ is a reachable state of $\fstype[\S]$ which belongs to
$P(q)$, i.e.\ $p\models \Chi{P(q)}$.

Furthermore, by definition of \REQs we know that $\emptyset \neq J \subseteq \N$ such that $\forall_{j\in J}\cdot a \in \Sigma^!_j \land a~\bkw{en}_{A^p_j}@q_j$
and
$\fstype^p(a)=(s,r) \Rightarrow |J|\in s \land 0\notin r$.
From $a~\bkw{en}_{A^p_j}@q_j$ for all $j\in J$ it follows, by definition of projection, that
in each $A_j$ there exists a transition $t_j$ starting in $q_j$ with action $a$ such that $p\models \gamma_j(t_j)$.
Therefore, $a~\bkw{en}_{A_j}@q_j$ for all $j\in J$ and $p \models \bkw{fe}(J,a,q)$.

Since $\fstype^p(a)=(s,r) \Rightarrow |J|\in s \land 0\notin r$ and $\fstype^p(a) =\fstype(p,a)$, we get that $p\in P(\fstype,J,a)$.
Thus $p \models \Chi{P(\fstype,J,a)}$.

In summary, $p \models \bkw{prod}(J,a,q)$ and
$\trFe{\bkw{prod}(J,a,q)}\,\ab\rcp(J,a)@q$ is an \FREQ for
$\fstype[\S]$.
\qed
\end{proof}

\lemfeaturedcompliance*
\begin{proof}
$(\bm\Rightarrow)$:
Since $\fstype[\S]$ is featured compliant,
there exists $n \geq 1$ and for $k=1,\ldots,n$
there exist transitions
$t^k = q\tr{(J,a,R^k)}_{\fstype[\S]} q^k$ with $R^k\neq\emptyset$
such that
$\models \psi \rightarrow \bigvee_{k\in \{1,\ldots,n\}} \gamma(t^k).$
Let $p\subseteq F$ be a product such that 
$p \models \psi$.
Then $p \models \bigvee_{k\in \{1,\ldots,n\}} \gamma(t^k).$
Hence, there exists $i\in \{1,\ldots,n\},\ab R^i \neq \emptyset$
and $q^i \in Q$ such that
$q\tr{(J,a,R^i)}_{\fstype[\S]} q^i$ and 
$p \models  \gamma(q\tr{(J,a,R^i)}_{\fstype[\S]}q^i)$.

\smallskip

$(\bm\Leftarrow)$:
Let $\{p_1,\ldots,p_n\}$ be the set of products which satisfy $\psi$.
This set is finite, since $F$ is finite, and nonempty, since $\psi$ is satisfiable.
For each $k=1,\ldots,n$ we can choose, by assumption,
a transition
$t^k = q\tr{(J,a,R^k)}_{\fstype[\S]} q^k$ with $R^k\neq\emptyset$
such that $p_k \models \gamma(t^k).$
Therefore $\models \psi \rightarrow \bigvee_{k\in \{1,\ldots,n\}} \gamma(t^k)$
and thus $\fstype[\S]$ is featured compliant with
$\trFe{\psi}\,\rcp(J,a)@q$. 
\qed
\end{proof}

\thmfrectorec* 
\begin{proof}
We perform the proof for receptiveness. The case of weak receptiveness is more technical but can be proven along the same lines (by using sequences of transitions and a straightforward generalisation of~\cref{lem:featured-compliance}). 

\smallskip

$(\bm\Rightarrow)$:
Let $p \in \sem{\fm}$ and let $\rcp(J,a)@q$ be an arbitrary \REQ 
for $\fstype[\S]^p$. Then, by~\cref{lem:featured-req}, 
$\trFe{\bkw{prod}(J,a,q)}\,\rcp(J,a)@q$ is an \FREQ 
for $\fstype[\S]$ and
$p \models {\bkw{prod}(J,a,q)}$.
By assumption, $\fstype[\S]$ is featured receptive
and thus featured compliant with this \FREQ.
Therefore, by ~\cref{lem:featured-compliance},
there exist
$R\neq\emptyset$ and $q' \in Q$ such that
$q\tr{(J,a,R)}_{\fstype[\S]} q'$ and $p \models \gamma(q\tr{(J,a,R)}_{\fstype[\S]}q')$.
Hence, by definition of \FETA projection,
$q\tr{(J,a,R)}_{\fstype[\S]^p} q'$.
This shows that $\fstype[\S]^p$ is compliant with $\rcp(J,a)@q$.
Since $\rcp(J,a)@q$ was chosen arbitrarily, 
$\fstype[\S]^p$ is receptive.

\smallskip

$(\bm\Leftarrow)$:
Let $\trFe{\bkw{prod}(J,a,q)}\,\rcp(J,a)@q$ be an arbitrary \FREQ 
for $\fstype[\S]$.
Let $p$ be an arbitrary product such that $p \models {\bkw{prod}(J,a,q)}$.
Then, by~\cref{lem:featured-req},
$p \in \sem{\fm}$ and $\rcp(J,a)@q$ is a \REQ for $\fstype[\S]^p$.
By assumption, $\fstype[\S]^p$ is compliant with this \REQ.
Hence, there exist $R\neq\emptyset$ and $q' \in Q$ such that
$q\tr{(J,a,R)}_{\fstype^p[\S]} q'.$
By definition of projection,
$q\tr{(J,a,R)}_{\fstype[\S]} q'$ and
$p \models \gamma(q\tr{(J,a,R)}_{\fstype[\S]}q')$.
Then, by~\cref{lem:featured-compliance},
$\fstype[\S]$ is featured compliant with $\trFe{\bkw{prod}(J,a,q)}\,\rcp(J,a)@q$.
Since the \FREQ was chosen arbitrarily, $\fstype[\S]$ is receptive. 
\qed
\end{proof}

\end{document}